\documentclass[11pt]{article}

\usepackage{enumitem}
\usepackage{amsthm}
\usepackage{amsmath}
\usepackage{amssymb}
\usepackage{graphicx}
\usepackage{fancyhdr}
\usepackage{mathabx}
\usepackage{hyperref}
\usepackage{algorithm}
\usepackage{algorithmic}
\usepackage{bbm}

\usepackage{geometry}

\newtheorem{theorem}{Theorem}
\newtheorem{lemma}[theorem]{Lemma}

\newtheorem{claim}[theorem]{Claim}
\newtheorem{corollary}[theorem]{Corollary}
\newtheorem{definition}[theorem]{Definition}
\theoremstyle{remark}
\newtheorem{remark}{Remark}
\theoremstyle{definition}

\DeclareMathOperator{\inv}{inv}
\DeclareMathOperator{\cost}{cost}

\DeclareMathOperator{\Err}{Err}

\title{Near-Optimal Bounds for Online Caching with Machine Learned Advice}
\author{Dhruv Rohatgi \\ MIT \\ drohatgi@mit.edu}

\begin{document}

\maketitle
\begin{abstract}
In the model of online caching with machine learned advice, introduced by Lykouris and Vassilvitskii, the goal is to solve the caching problem with an online algorithm that has access to next-arrival predictions: when each input element arrives, the algorithm is given a prediction of the next time when the element will reappear. The traditional model for online caching suffers from an $\Omega(\log k)$ competitive ratio lower bound (on a cache of size $k$). In contrast, the augmented model admits algorithms which beat this lower bound when the predictions have low error, and asymptotically match the lower bound when the predictions have high error, even if the algorithms are oblivious to the prediction error. In particular, Lykouris and Vassilvitskii showed that there is a prediction-augmented caching algorithm with a competitive ratio of $O(1+\min(\sqrt{\eta/\textsc{opt}}, \log k))$ when the overall $\ell_1$ prediction error is bounded by $\eta$, and $\textsc{opt}$ is the cost of the optimal offline algorithm.

The dependence on $k$ in the competitive ratio is optimal, but the dependence on $\eta/\textsc{opt}$ may be far from optimal. In this work, we make progress towards closing this gap. Our contributions are twofold. First, we provide an improved algorithm with a competitive ratio of $O(1 + \min((\eta/\textsc{opt})/k, 1) \log k)$. Second, we provide a lower bound of $\Omega(\log \min((\eta/\textsc{opt})/(k \log k), k))$.
\end{abstract}

\section{Introduction}

In the \emph{online caching} problem (also known as \emph{paging}), we are given a sequence of elements which arrive one at a time, and we must maintain a cache of some fixed size $k$. The cost of a caching algorithm on some input is the number of cache misses. The standard goal is to design an online algorithm with minimal competitive ratio, relative to the optimal offline algorithm.

As a fundamental problem in the study of online algorithms, caching has been extensively studied \cite{Sleator1985, McGeoch1991, Fiat1991, Achlioptas2000}; it is well-known that the optimal competitive ratio of any deterministic algorithm is $k$, and the optimal competitive ratio of any randomized algorithm is $2H(k) - 1 = \Theta(\log k)$, where $H(k)$ is the $k$-th harmonic number \cite{Achlioptas2000}.

However, the traditional framework for analyzing online algorithms---namely, worst-case competitive ratios---is overly pessimistic, by virtue of requiring worst-case analyses. Real-world data often satisfies nice properties---it may be predictable, or simply random, or even just not adversarial---and for this reason, theoretically unsound algorithms can perform very well in practice. Numerous attempts have been made to theoretically ground this observation; some of the more prominent are average-case analyses \cite{KP2000, Ajwani2007, McGregor2014} and smoothed analyses \cite{Spielman2004, Spielman2009}, both within online algorithms and beyond.

One such attempt which has recently garnered significant attention is the framework of online algorithms with machine learned advice. In this model, the online algorithm is augmented with an oracle that makes certain predictions about future data. In practice, this oracle is likely to be a machine learned predictor. Since machine learning is imperfect (and can sometimes be wildly wrong), the algorithm must incorporate the oracle's advice judiciously, without being given any bound on the oracle's error. The goal is to develop an algorithm which is both \emph{consistent}, in that it nearly matches the best offline algorithm when the predictor is nearly perfect, as well as \emph{robust}, in that its worst-case performance is good even when the oracle is arbitrarily bad. The performance of the algorithm should be bounded as a function of some measure of the oracle error, even though the algorithm is oblivious to this error.

The ML advice model has in the past been applied to the ski rental problem \cite{Purohit2018, Gollapudi2019}, job scheduling \cite{Purohit2018, Mitzenmacher2019} and online revenue maximization \cite{MV2017}; it has also been used to achieve theoretical and practical gains in streaming frequency estimation \cite{Hsu2018} and data structures \cite{Kraska2018}. Most relevant to this paper is prior work by \cite{LV2018} in which it was shown how the model can be applied to the online caching problem. In particular, they considered augmenting caching algorithms with an oracle that predicts the next arrival of each element. The oracle's $\ell_1$ error is defined as the sum over all elements of the absolute difference between the element's true and predicted next arrival. In this model, they developed a ``predictive marker algorithm'' with the following guarantee:

\begin{theorem}\cite{LV2018}\label{theorem:lv}
The predictive marker algorithm achieves a competitive ratio of $\min(2 + 4\sqrt{\eta/\textsc{opt}}, 4H(k))$ when the oracle has $\ell_1$ error of at most $\eta$, and the cost of the optimal offline algorithm is $\textsc{opt}$.
\end{theorem}

Note that the competitive ratio achieved is $O(\min(\sqrt{\eta/\textsc{opt}}, \log k))$. This ratio is of course (asymptotically) optimal as a function of $k$. However, it is an open question how far the dependence on $\eta/\textsc{opt}$ (a measure of relative error, in some sense) can be improved. In particular, it would be interesting to understand how accurate the predictions need to be in order to provide an improvement in the competitive ratio. From the previous work, it is only shown that $\eta/\textsc{opt} = o(\log^2 k)$ suffices.

In this paper, we work with the same model, and make progress on this question. Building upon the techniques used in \cite{LV2018}, we provide an algorithm with an improved competitive ratio:

\begin{theorem}\label{theorem:nonmarker}
There is an algorithm for caching with predictions that achieves a competitive ratio of $$O\left(1 + \min\left(1, \frac{\eta/\textsc{opt}}{k}\right) \log k \right)$$ when the oracle has $\ell_1$ error of at most $\eta$.
\end{theorem}

Our bound matches the prior work when $\eta/\textsc{opt} \geq k$ (in that neither algorithm improves upon the classical $O(\log k)$ competitive ratio) and is strictly better when $\eta/\textsc{opt} < k$. For example, if $\eta/\textsc{opt} = k/\log k$, then the prior algorithm had a competitive ratio of $O(\log k)$, whereas we show that $O(1)$ is possible.

Furthermore, we provide a lower bound, stated informally as follows:

\begin{theorem}
Any randomized algorithm for caching with predictions must have a competitive ratio which is $$\Omega\left(\log \min\left(\frac{\eta/\textsc{opt}}{k\log k}, k\right)\right)$$ as a function of $\eta/\textsc{opt}$ and $k$.
\end{theorem}

To our knowledge, this is the only known lower bound. The upper bound and lower bound are asymptotically tight when $\eta/\textsc{opt} \leq k/\log k$ or $\eta/\textsc{opt} \geq k^{1+\epsilon}$. There is a still a significant gap; the two bounds are non-trivial on disjoint regimes, and in the regime $k \leq \eta/\textsc{opt} \leq k \log k$, neither bound is non-trivial.

Nonetheless, these results make progress towards determining the largest possible error bound that still admits a non-trivial competitive ratio. Where prior work only showed that $\eta/\textsc{opt} = o(\log^2 k)$ suffices to obtain a competitive ratio of $o(\log k)$, the above results imply that $\eta/\textsc{opt} = o(k)$ suffices and $\eta/\textsc{opt} \leq k^{1 + o(1)}$ is necessary.

\subsection{Roadmap}

In Section~\ref{section:prelim}, we formally describe the online caching model with machine learned advice, and define the predictor. We also review some facts from traditional caching algorithms (specifically, facts about marker-based algorithms) that we will rely on later in the paper.

Having defined the necessary terminology, we then outline in Section~\ref{section:technique} the techniques used to achieve our results.

In Section~\ref{section:algo1}, we provide the first of these results, a warm-up algorithm for caching with predictions. This algorithm improves upon the prior work, and is simpler than our final algorithm (and thus may have some practical advantages). Furthermore, as a marker-based algorithm, this first algorithm is somewhat simpler to analyze.

In Section~\ref{section:algo2}, we describe our second and final algorithm, departing from the marker-based framework to achieve an improvement in the competitive ratio.

Finally, in Section~\ref{section:lbound}, we prove the lower bound.

\section{Preliminaries}\label{section:prelim}

\subsection{Traditional caching}

In the traditional online caching problem, the input is a sequence $\sigma = (z_1,z_2,\dots,z_n)$ of elements which become available one by one. The cache has fixed size $k$ and is initially empty. As elements arrive, if an element is not present in the cache, then it counts as a ``cache miss'', and the algorithm must add it to the cache, and choose which cache element to evict. Otherwise nothing happens. The cost $\text{cost}_A(\sigma)$ of the algorithm $A$ on the input $\sigma$ is the number of cache misses. If the algorithm is randomized, this cost is the expected number of cache misses.

We define $\textsc{opt}(\sigma)$ to be the minimum number of caches misses achievable by an ``offline'' algorithm---an algorithm which is given $\sigma$ in advance. Our online algorithm $A$ is $\alpha$-competitive if there is some constant $c$ such that for every input $\sigma$, $$\text{cost}_A(\sigma) \leq \alpha \cdot \textsc{opt}(\sigma) + c.$$

It is known that there is a $k$-competitive deterministic algorithm and an $O(\log k)$-competitive randomized algorithm. Furthermore, these ratios are optimal.

\subsection{ML advice}

In this paper we consider not the traditional model but rather an extension of it, in which our algorithm is also given some advice \cite{LV2018}. In particular, when input element $z_i$ arrives, an oracle gives the algorithm $h_i(z_i)$, which is an estimate of $y_i = \min_{j > i} \{j: z_j = z_i\}$, the next time when element $z_i$ will appear (if $z_i$ never appears again, and the input sequence has length $n$, then we set $y_i = n+1$). These estimates may not be correct, and we want to bound the performance of our algorithm as a function of the error. We define the error as the $\ell_1$ distance between the real and predicted next arrivals: for each input element $z_i$ we define $\text{Err}_i(z_i) = |h_i(z_i) - y_i|$, and then define $$\eta = \sum_{i=1}^n \text{Err}_i(z_i).$$

For any time $i$ and input element $w$, we also define $L(w,i)$ to be the last time $j < i$ such that $z_j = w$. 


When analyzing an algorithm $A$ in this model, the goal is to bound its competitive ratio as a function of $\eta/\textsc{opt}$: more precisely, to show for a desired function $\alpha$ and a constant $c$, that $\text{cost}_A(\sigma) \leq \alpha(\eta/\textsc{opt}) \cdot \textsc{opt}(\sigma) + c$ for every input $\sigma$. This is the approach taken in prior work (see Theorem~\ref{theorem:lv}), and we will see how $\eta/\textsc{opt}$ arises naturally in our algorithms' analyses. It would not make sense for the competitive ratio to be a function of the absolute error $\eta$, since duplicating the input sequence would double $\eta$ but leave the ratio $(\text{cost}_A(\sigma)/\textsc{opt}(\sigma))$ approximately unchanged.

\subsection{Marker-based algorithms}

Our first predictive caching algorithm will be a \emph{marker-based algorithm} which judiciously incorporates the oracle's advice; our second algorithm will depart from but still rely heavily on the marker-based framework. Marker-based caching algorithms have the following structure. The execution of the algorithm comes in phases, and at the beginning of each phase, every cache element is said to be \emph{unmarked}. When a cache hit occurs, the corresponding element is \emph{marked}. When a cache miss occurs, some unmarked element is evicted from the cache, and the new element is inserted in its place, and immediately marked. If all elements of the cache are marked and another cache miss occurs, then the whole cache is unmarked and a new phase begins.

The decision that a marker-based algorithm has to make is which unmarked element to evict in the event of a cache miss. In the traditional online model, the randomized marker algorithm achieves an $O(\log k)$-competitive ratio by evicting a random unmarked element. Additionally, any marker-based algorithm is $k$-competitive.

For any marker-based algorithm, we can make the following definition.

\begin{definition}
An input element is called \emph{clean} for some phase $r$ of the algorithm execution if it appeared in phase $r$ but did not appear in phase $r-1$. If it appeared in both phase $r$ and phase $r-1$, it is called \emph{stale}.
\end{definition}

The following lemma is known, relating the number of clean elements to the optimal offline cost.

\begin{lemma}\cite{Fiat1991}\label{lemma:fiat}
Let $L$ be the number of clean elements in an execution of a marker-based algorithm on some input $\sigma$. Then $L/2 \leq \textsc{opt}(\sigma) \leq L$.
\end{lemma}

The phases (and consequently, the clean/stale elements) are in fact independent of the exact algorithm.

\begin{definition}
An \emph{arrival} in phase $r$ is an element $z_i$ which has not previously appeared in the same phase.
\end{definition}

Then the following fact can be readily derived for any marker-based algorithm:

\begin{claim}
Every phase contains exactly $k$ arrivals.
\end{claim}

More specifically, each phase continues as long as possible without containing $k+1$ distinct elements.

\subsection{Eviction chains}

To design and analyze marker-based algorithms, it is useful to decompose the set of cache misses into \emph{eviction chains}, a concept perhaps first explicitly utilized in \cite{LV2018}.

For any marker-based algorithm and any phase, each clean arrival in the phase causes a cache miss. This yields a chain of evictions in that phase which can be blamed on that clean arrival: the clean element's arrival evicts some element, whose next appearance evicts another element, and so forth until an element is evicted which never reappears in the phase. Each element in the chain must be an arrival, since elements which have previously appeared were marked and are therefore immune to eviction for the remainder of the phase. Thus, each element in the chain after the first clean element must be stale, since to be evicted it must have been present in the cache.

These clean-element chains account for all cache misses in the phase, since every stale element was in the cache at the start of the phase, so for it to cause a cache miss it must have been evicted by a previous element. So the total number of cache misses in a phase is simply the total length of the eviction chains.

\section{Our techniques and related work}\label{section:technique}

Given next-arrival predictions, an algorithmically naive approach is to trust the predictions completely. The optimal offline algorithm evicts, at each cache miss, the cache element with the latest next-arrival time. Thus, if the predictions are perfect then this approach will have a competitive ratio of $1$. However, even small errors in the predictions can lead to an unbounded prediction. So it is necessary to balance trusting the predictions with making provably competitive decisions.

Our work builds on \cite{LV2018}, which proposed a marker-based algorithm for caching with predictions. Their algorithm utilizes eviction chains. Since eviction chains partition the set of cache misses, and the number of eviction chains is equal to the number of clean arrivals, which is asymptotically equal to $\textsc{opt}$, bounding the average chain length bounds the competitive ratio of the algorithm. Thus, algorithms which work with each chain independently can often be cleanly analyzed. In \cite{LV2018}, this approach is carried out: for each eviction chain, the predictions are trusted (to choose which element to evict next) until the chain reaches length $\Omega(\log k)$, after which evictions are random. Each chain's length can be bounded by the prediction error of elements in that chain. Since the chains are disjoint, this implies a bound on the total cost of the algorithm by the total prediction error (see Theorem~\ref{theorem:lv}).

As a warm-up, we first show that a small modification to the algorithm from \cite{LV2018} improves the competitive ratio from $O(1+\min(\sqrt{\eta/\textsc{opt}}, \log k))$ to $O(1+\min(\log (\eta/\textsc{opt}), \log k))$. The modification is simple---trust the predictions only \emph{once} in each chain---but the analysis requires more care. Unlike before, the length of each chain now depends on the next-arrival prediction errors of elements which may or may not appear in the chain. Thus, adding up the errors could double-count. To avoid this issue, we do not directly bound each chain's length by prediction error. Instead, we charge the length of each chain against a set of inversions in the order of element arrivals relative to the predictions. These sets of inversions are disjoint for different chains, so summing across chains does not double count. A combinatorial lemma then relates the total number of inversions to the total error.

To improve the competitive ratio further and achieve the bound stated in Theorem~\ref{theorem:nonmarker}, we depart from the marker-based framework. In our previous algorithm, trusting the prediction once meant evicting the unmarked element $e$ with latest predicted next arrival. Intuitively, if the prediction error ``per chain'' is $O(\eta/\textsc{opt})$, then if $e$ reappears in the same phase, it should on average be at most $O(\eta/\textsc{opt})$ elements away from the end of the phase. Our first algorithm would then proceed by random eviction of unmarked elements, so $e$'s eviction chain would have length $O(\log (\eta/\textsc{opt}))$ in expectation (for the same reason that the traditional marker algorithm has competitive ratio $O(\log k)$).

However, if $\eta/\textsc{opt} \ll k$, then a uniformly random cache element (potentially marked) will probably not appear after $e$, so evicting it might terminate the eviction chain at length $O(1)$ instead of $O(\log (\eta/\textsc{opt}))$. This is the motivation for our final algorithm. However, it relies on evicting marked cache elements, which significantly complicates the analysis, since facts about marker-based algorithms no longer directly apply.

In particular, the number of eviction chains may no longer be $\Theta(\textsc{opt})$; evicting a marked element may cause an ``extra'' eviction chain in the next phase. To deal with this issue, our key tool is a bijection from eviction chains to special elements that began the phase in the cache but never appeared. With this bijection, we show that every extra chain can be charged against either prediction error or a chain which \emph{does not} exist but could have. An added complication is showing that these prediction errors are disjoint.

For the lower bound, the idea is to choose an input distribution and predictions such that the predictions give no information about the future input, but have reasonably small error. Generalizing the traditional lower bound against randomized caching algorithms---where each input element is uniformly distributed over $k+1$ pages, so each phase has one clean element---our input distribution consists of phases each with $t$ clean elements, where $t$ is a variable parameter of the distribution. Increasing the parameter $t$ allows smaller prediction error, at the cost of a smaller bound on the competitive ratio.

\section{Marker-based predictive algorithm}\label{section:algo1}

Our first algorithm \textsc{lmarker} is a modification of the algorithm proposed in \cite{LV2018}, which is itself a balance between two paradigms: trust the oracle, or ignore the oracle. A marker-based algorithm which trusts the oracle would always evict the cache element with highest predicted next arrival; if the oracle had zero error, this algorithm would match the optimal offline algorithm. A marker-based algorithm which ignores the oracle is the random marker algorithm. To combine the gains of the former algorithm when the oracle has low error with the robustness of the latter algorithm when the oracle has high error, the algorithm from \cite{LV2018} uses the following strategy: for each eviction chain, trust the oracle until the chain becomes long, and subsequently ignore the oracle.

The modification we make is simple: trust the oracle less---only once at the beginning of each eviction chain. Intuitively, a total prediction error of $\eta$ translates to an error of $\Theta(\eta/\textsc{opt})$ in each of the $\Theta(\textsc{opt})$ chains. Evicting the unmarked element with highest predicted next arrival means that its true next arrival should be only $O(\eta/\textsc{opt})$ from the end of the phase, resulting in $O(\log \eta/\textsc{opt})$ more cache misses for that eviction chain.

We now describe \textsc{lmarker} in more detail, and formalize the analysis. The chains cannot quite be analyzed independently (unlike in \cite{LV2018}), so care is needed.

\paragraph{Algorithm description} \textsc{lmarker} is a marker-based algorithm with the following eviction strategy upon a cache miss: if the incoming element is clean, then evict the unmarked element with the highest predicted arrival time. If the incoming element is stale, then evict a random unmarked element.

\paragraph{Algorithm analysis} Fix a phase. Exactly $k$ distinct elements arrive during the phase; let $i_1,\dots,i_k$ be the times of the first arrivals for the phase. These are the only times when cache misses may occur, and the only times when an unmarked element becomes marked.

Consider a single clean element with arrival time $i_t$. It evicts the unmarked element with the highest predicted arrival time, which must be either (1) a stale element with arrival time in the set $\{i_{t+1},\dots,i_k\}$ or (2) an element which does not appear in this phase. Case (2) results in no more cache misses along this chain, so we analyze case (1): some stale element is evicted. Let $i_{e(t)}$ be the time at which it arrives.

For any $1 \leq a \leq b \leq k$, define $N_a(b)$ to be the number of stale elements which are unmarked and in the cache at time $i_a$, and have arrival times after $i_b$. Let $\mathcal{E}_t$ be the distribution over executions of the algorithm up to time $t$.

\begin{claim}\label{claim:randomunmarked}
If the clean element arriving at time $i_t$ evicts the stale element arriving at time $i_{e(t)}$, then the expected length of the eviction chain begun at $i_t$ is at most $O(\mathbb{E}[1 + \log(N_{t}(e(t)))])$.
\end{claim}

\begin{proof}
Conditioning on $\mathcal{E}_t$, the random variable $N_{t}(e(t))$ is determined. In the worst case, each stale element evicts another stale element until this is no longer possible. There are at most $N_{t}(e(t))$ unmarked stale cache elements at time $i_{e(t)}$. Say that there are $m$ such elements, and order them $1,\dots,m$ by arrival time. If $z_{i_{e(t)}}$ evicts the $j$-th such element then there will be at most $m-j$ unmarked cache elements when that element arrives. But $j$ is uniformly distributed over $\{1, 2, \dots, m\}$. Thus, the expected length of the chain, conditioned on $\mathcal{E}_t$, is bounded by $R_{N_{t}(e(t))}$ defined by the recurrence $$R_m = 1 + \frac{1}{m} \sum_{j=1}^m R_{m-j},$$ which solves to $R_m = O(\log m)$. By the law of total expectation, the unconditional expected length of the chain is bounded by $O(\mathbb{E}[1 +  \log(N_t(e(t)))])$.
\end{proof}

Next, fix any execution of the algorithm on the entire input. We relate $N_t(e(t))$ to prediction error; more specifically, we relate it to the number of inversions in the predicted arrival order of stale elements in phase $r$. For each stale arrival time $i_s$, let $j_s$ be the most recent appearance of $z_{i_s}$ in phase $r-1$. Let $J = \{j_s | z_{i_s}\text{ is stale in phase $r$}\}$ be the set of most recent appearances.

If $z_{i_{e(t)}}$ is evicted by $z_{i_t}$, then at time $i_t$, all unmarked stale elements $w$ in the cache with arrivals in the set $\{i_{e(t)+1},\dots,i_k\}$ had earlier predicted next arrivals than $z_{i_{e(t)}}$. Thus, the number of $u > e(t)$ for which $z_{i_u}$ is stale and $h_{j_{e(t)}}(z_{j_{e(t)}}) \geq h_{j_u}(z_{j_u})$ is at least $N_t(e(t))$.

Define $N = \sum N_t(e(t))$, summing over all clean arrival times $i_t$. Then the sequence of predicted stale arrivals $(h_j(z_j))_{j \in J}$ has at least $N$ inversions with respect to the strictly increasing integer sequence of actual stale arrivals. It follows from Lemma~\ref{lemma:combo}, which we state and prove in the next section, that $$\sum_{j_s \in J} |h_{j_s}(z_{j_s}) - i_s| \geq N/2.$$

Let $\eta_{r-1}$ be the sum of prediction errors over all predictions made in phase $r-1$. Since $J$ is a set of times in phase $r-1$, we get from the above inequality that $\eta_{r-1} \geq N/2$. On the other hand, by Claim~\ref{claim:randomunmarked}, the expected number of cache misses in phase $r$ is at most $\mathbb{E}\left[\sum_t O(1 + \log N_t(e(t)))\right]$, summing over all $t$ such that $i_t$ is a clean arrival time. 

To simplify notation a bit, for each clean arrival time $c = i_t$, define $N_c = N_t(e(t))$. Then we have shown that $2\eta_{r-1} \geq \sum_{c \in r} N_c$, where the sum is over clean arrivals $c$ in phase $r$. Furthermore, we have shown that $$\mathbb{E}[\text{\# of cache misses in phase $r$}] \leq \sum_{c \in r} O(\mathbb{E} [1 + \log(N_c)])  \leq \sum_{c \in r} O(1 + \log(\mathbb{E} [N_c]))$$ by Jensen's inequality.

Now sum over all phases. The number of cache misses is $O(L + \sum_c \log(\mathbb{E} [N_c]))$, where the sum is over the $L$ clean arrivals in all phases. And we know that $2 \eta \geq \sum_c N_c$ for all executions, implying that $2\eta \geq \sum_c \mathbb{E} [N_c]$. By another application of Jensen's inequality, the expected number of cache misses can be bounded by $O(L + L\log(2\eta/L))$.

Since we also know that $N_c \leq k$ for any clean arrival $c$, we can alternately bound the number of cache misses by $O(L + L \log k)$.

Combining the two bounds and using the fact that $L = \Theta(\textsc{opt})$, we have proven the following theorem.

\begin{theorem}
The algorithm \textsc{lmarker} has competitive ratio $O(1 + \min(\log (\eta/\textsc{opt}), \log k))$, where $\eta$ is the unknown $\ell_1$ prediction error.
\end{theorem}

Let $H(m)$ be the $m$-th harmonic number. We can keep track of the exact constants in the proof, rather than using asymptotic notation:

\begin{theorem}
The algorithm \textsc{lmarker} has competitive ratio $4 + 2H(\min(2\eta/\textsc{opt}, k))$, where $\eta$ is the unknown $\ell_1$ prediction error.
\end{theorem}

\subsection{Proof of combinatorial lemma}

Let $M = (m_1,\dots,m_n)$ be a strictly increasing integer sequence. For any integer sequence $A = (a_1,\dots,a_n)$ let $\inv(A)$ be the number of pairs of indices $i<j$ such that $a_i \geq a_j$. Let $\cost(A) = \sum_{i=1}^n |a_i - m_i|$, and define $\Delta(A) = 2\cost(A) - \inv(A)$.

\begin{lemma}\label{lemma:combo}
Let $A = (a_1,\dots,a_n)$ be an arbitrary integer sequence. Then $\inv(A) \leq 2\cost(A)$, with $\inv(A)$ and $\cost(A)$ as defined above.
\end{lemma}

\begin{proof}
Without loss of generality we can assume that all elements of $A$ are bounded between $m_1$ and $m_n$, since outliers can be thresholded without decreasing $\inv(A)$ or increasing $\cost(A)$. So the set of sequences is finite. Let $B = (b_1,\dots,b_n)$ be a sequence which minimizes $\Delta$, and assume that $B$ has the maximum possible sum of elements out of all sequences minimizing $\Delta$.

Suppose that $B$ were to have a strict inversion. Then there is some $1 \leq i < n$ such that $b_i > b_{i+1}$. Define two sequences $B^l$ and $B^h$ so that $b^l_i = b_{i+1}$ and $b^h_{i+1} = b_i$, and in all other locations $B^l$ and $B^h$ agree with $B$. Then by construction, $$\inv(B) - \inv(B^l) = \inv(B^h) - \inv(B).$$ Furthermore,
\begin{align*}
\cost(B) - \cost(B^l)
&= |b_i - m_i| - |b_{i+1} - m_i| \\
&\geq |b_i - m_{i+1}| - |b_{i+1} - m_{i+1}| \\
&= \cost(B^h) - \cost(B)
\end{align*}
since $m_i < m_{i+1}$, and the function $|b_i - x| - |b_{i+1} - x|$ is non-increasing. It follows from the above inequalities and the optimality of $B$ that $$\Delta(B^h) - \Delta(B) \leq \Delta(B) - \Delta(B^l) \leq 0.$$ Therefore $B^h$ minimizes $\Delta$ as well. But it has a larger sum of elements than $B$. Contradiction.

So we know that $B$ is non-decreasing, and thus $\inv(B)$ is exactly the number of pairs of equal elements. But a constant sequence of length $l$ has cost at least $\binom{l}{2}/2$, and contributes only $\binom{l}{2}$ pairs of equal elements. Partitioning $B$ into constant sequences, we get $\inv(B) \leq 2\cost(B)$.
\end{proof}

\section{Improved algorithm}\label{section:algo2}

In the previous section, we saw that using the predictions once at the beginning of each eviction chain takes the chain most of the way through the phase; random evictions of unmarked elements are then used until the chain ends. Suppose that the second element of an eviction chain does not reappear until $O(\eta/\textsc{opt})$ steps away from the end of the phase. Then evicting a uniformly random element of the cache---marked or unmarked---would terminate the chain immediately with probability $1 - O((\eta/\textsc{opt})/k)$.

In this section we present an improved algorithm \textsc{lnonmarker} motivated by the above observation. As the name may suggest, it is not quite a marker-based algorithm, and we need to give some new names to familiar concepts:

\begin{definition}
For any input sequence $z$ and cache size $k$, the \emph{phases} of the input sequence are defined recursively as follows: phase $r$ begins right after the end of phase $r-1$, and extends as long as possible without containing $k+1$ distinct elements.
\end{definition}

\begin{definition}
Fix an algorithm. An input element is called \emph{initial} for some phase $r$ if it appeared in phase $r$, and was present in the cache at the beginning of phase $r$. If it appeared in phase $r$ but was not present in the cache at the beginning of the phase, it is called \emph{non-initial}. 
\end{definition}

Note that the definition of phases given here coincides with the phases of any marker-based algorithm. This definition is algorithm-independent, and thus is also useful for non-marker-based algorithms. For any marker-based algorithm, the definitions of \emph{clean} and \emph{non-initial} coincide, as do the definitions of \emph{stale} and \emph{initial}. However, it can be seen that these may diverge in the execution of a non-marker-based algorithm. Some facts about clean and stale elements which we used in the previous section are now facts about initial and non-initial elements:

\begin{claim}\label{claim:nonmarkerfacts}
Every phase contains exactly $k$ arrivals. Every non-initial arrival causes a cache miss. Every other cache miss in the phase was caused by some previous cache miss in the same phase.
\end{claim}

Thus, in analogy with clean arrivals, every non-initial arrival heads an eviction chain, and the eviction chains partition the set of cache misses. However, unlike before, not all cache misses are necessarily upon arrivals: an element might arrive, be evicted, and reappear all in one phase.

\paragraph{Algorithm description} The new algorithm \textsc{lnonmarker} still maintains markings on cache elements. At the beginning of each phase, all cache elements are unmarked. Whenever a cache hit occurs, the element is marked. Whenever a cache miss occurs, the algorithm evicts some element and marks the new element. In particular, the algorithm has the following eviction strategy upon a cache miss: if the incoming element was non-initial, then evict the unmarked element with the highest predicted arrival time. If the incoming element was evicted by a non-initial element, then evict a uniformly random element of the cache (not necessarily unmarked). In all other cases, evict a uniformly random unmarked element.

Before analyzing the algorithm, we must show that it is in fact well-defined.

\begin{claim}
At any cache miss, there is at least one unmarked element. Thus, the algorithm is well-defined.
\end{claim}

\begin{proof}
Each marked element of the cache must have appeared in the current phase. The element which caused the cache miss is, of course, distinct from all elements of the cache. Since the phase contains at most $k$ distinct elements, the cache contains at most $k-1$ marked elements.
\end{proof}

\paragraph{Algorithm analysis} Fix a phase $r$ in which the set of elements is $A$. Fix a single execution of the algorithm. Let $S$ be the cache at the beginning of phase $r$. Then $A \setminus S$ is the set of non-initial elements, and $S \setminus A$ is the set of cache elements which do not arrive in phase $r$. Two facts relate these sets:

\begin{itemize}
\item For each non-initial element $c$, the corresponding eviction chain evicts at most one element of $S \setminus A$. 
\item Every $x \in S \setminus A$ is evicted at most once in the phase: $x$ does not appear in the phase, so after it is evicted it will not return to the cache.
\end{itemize}

Together with the observation that $|A| = |S| = k$, these facts imply that there is a bijection $f: A \setminus S \to S \setminus A$ such that if $c$'s chain evicts $x \in S \setminus A$ then $f(c) = x$. For any $c \in A \setminus S$ such that $c$'s chain does not evict any element of $S \setminus A$, we set $f(c)$ arbitrarily, subject to the condition that $f$ is a bijection.

By Claim~\ref{claim:nonmarkerfacts}, the number of cache misses in phase $r$ is the total length of the eviction chains headed by the non-initial elements of phase $r$. Let $c$ be one such non-initial element, with eviction chain of length $\text{length}(c)$, which arrives at time $t_c$ and evicts some unmarked element $e$. Suppose $e$ reappears in the phase. Let $N(c)$ be the number of first arrivals after $e$ in the same phase, and let $N^*(c)$ be the number of distinct elements after $e$ (not necessarily first arrivals).

The above definitions were made for a single execution, but the algorithm defines a distribution $\mathcal{E}$ over executions. Let $\mathcal{E}_r$ be the distribution of executions of the first $r-1$ phases, and let $\mathcal{E}_{t_c}$ be the distribution of executions through time $t_c$.

Conditioned on $\mathcal{E}_r$, the non-initial elements $c$ are determined, but $N^*(c)$ and $N(c)$ are random variables (over the randomness of the algorithm). Defining $\text{length}(c) = 0$ and $N^*(c) = 0$ if $c$ is not a non-initial element or if $e$ does not reappear in the phase, we get the following result.

\begin{claim}\label{claim:chainlength}
For any element $c$ in phase $r$, $$\mathbb{E}_\mathcal{E}[\text{length}(c) | \mathcal{E}_r] \leq \alpha \cdot \mathbb{E}_\mathcal{E}\left[1 +  \frac{N^*(c)}{k}\log N(c) \middle| \mathcal{E}_r\right]$$ for an absolute constant $\alpha$.
\end{claim}

\begin{proof}
Condition on $\mathcal{E}_{t_c}$. Then $e$ is determined, and $N^*(c)$ and $N(c)$ are determined. If $e$ does not reappear in the phase, then $\mathbb{E}[\text{length}(c)|\mathcal{E}_{t_c}] = 0$. Suppose otherwise. Since $e$ evicts a random element $g$ from the cache, the probability that $g$ subsequently appears during the same phase is at most $p = N^*(c)/k$. If $g$ does appear, the eviction chain continues by random eviction of unmarked elements, until an evicted element is not in $A$. The number of unmarked elements at the time of $g$'s cache miss which are present in $A$ is at most $N(c)$, since any element which arrived earlier in the phase is either marked or no longer in the cache. Thus, by the same argument as in Claim~\ref{claim:randomunmarked}, the expected length of the chain is $O(\log N(c))$, conditioned on $g$'s cache miss. Since the length is $O(1)$ if $g$ does not appear, it follows that $\mathbb{E}[\text{length}(c)|\mathcal{E}_{t_c}]$ is at most $O\left(1 + p \log N(c)\right)$. Taking the expectation over $\mathcal{E}_{t_c}|\mathcal{E}_r$ yields the claimed result.
\end{proof}

Now we want to bound $N^*(c)$ in terms of the predictor error. We condition on the entire execution of the algorithm: that is, the following lemmas hold deterministically for all executions.

\begin{lemma}\label{lemma:lengthtoerror}
For any chain $(c, e, \dots)$ in which $e$ reappears after eviction, let $t$ be $c$'s arrival time. Let $r'$ be the next phase in which $f(c)$ arrives. Then $$\Err_{L(f(c),t)}(f(c)) + \Err_{L(e,t)}(e) \geq N^*(c) + k(r' - r - 1).$$
\end{lemma}

\begin{proof}
Note that $f(c)$ is an unmarked element of the cache at time $t$. It does not appear in phase $r$, so if phase $r+1$ begins at time $t_{r+1}$ then $f(c)$ does not appear until at least time $t_{r+1} + k(r' - r - 1)$, but its predicted next appearance $h_{L(f(c),t)}(f(c))$ satisfies $h_{L(f(c),t)}(f(c)) \leq h_{L(e,t)}(e)$, since $e$ was the unmarked cache element which maximized predicted arrival at time $t$. The next appearance of $e$ is no later than time $t_{r+1} - N^*(c)$, since at least $N^*(c)$ elements lie between $e$'s appearance and the end of the phase. So $\Err_{L(f(c),t)}(f(c)) + \Err_{L(e,t)}(e) \geq N^*(c) + k(r'-r-1)$.
\end{proof}

\begin{remark}
There is an edge case in Lemma~\ref{lemma:lengthtoerror} if $f(c)$ never appears after phase $r$ and the last phase has less than $k$ elements. In this case the inequality in Lemma~\ref{lemma:lengthtoerror} can be off by as much as $k$. However, when the lemma is applied in Lemmas~\ref{lemma:totalerror} and \ref{lemma:injection}, the applications almost always implicitly weaken the inequality by at least $k$ anyway. Each application where the inequality is not weakened corresponds to a chain in one of the last two phases. There are only $O(k)$ such chains, each causing a discrepancy of at most $k$. To make the lemmas hold, it therefore suffices to replace $\eta$ by $\eta + O(k^2)$. However, $O(k^2)$ is bounded as the input length and $\textsc{opt}$ grow, so the competitive ratio is unaffected.
\end{remark}

Next we would like to sum the inequality given by Lemma~\ref{lemma:lengthtoerror} across all chains in all phases, to bound $\sum N^*(c)$ in terms of the total prediction error. Within a single phase, $f(c)$ and $e$ each uniquely determine $c$, so the error of each prediction is counted at most once. However, the error of a single prediction may be counted in multiple chains in successive phases. The following lemmas do a more careful summation, taking into consideration the double-counting.

\begin{lemma}\label{lemma:eunique}
For distinct times $t_1, t_2$ suppose $c_1 = z_{t_1}$ and $c_2 = z_{t_2}$ are non-initial elements of phases $r_1$ and $r_2$, which evict $e_1$ and $e_2$ respectively. Suppose that $e_1$ and $e_2$ reappear after eviction. Then $L(e_1, t_1) \neq L(e_2, t_2)$.
\end{lemma}

\begin{proof}
If equality were to hold, then $e_1 = e_2 = w$ for some $w$. Note that $w$ is evicted at $c_1$'s arrival, and $w$ is evicted at $c_2$'s arrival, but $w$ is unmarked both times. So the evictions occurred in different phases. Without loss of generality suppose that $r_1 < r_2$. Then $w$ appeared in phase $r_1$ after eviction by $c_1$, so $L(w, t_1) < L(w, t_2)$.
\end{proof}

\begin{lemma}\label{lemma:totalerror}
Summing over all chains $(c,e,\dots)$ in which $e$ reappears after eviction, $$\sum_c N^*(c_i) \leq 3\eta$$ where $\eta$ is the total prediction error.
\end{lemma}

\begin{proof}
Fix a time $t$ and consider the set of chains $(c_i, e_i, \dots)$ such that $L(f(c_i), t_i) = t$, where $t_i$ is the arrival time of $c_i$: that is, $f(c_i) = z_t$ and the most recent appearance of $f(c_i)$ prior to time $t_i$ was at time $t$. Each chain is in a different phase, since $f$ is injective for any one phase. If there are $m_t$ such chains in phases $r_1 < \dots < r_{m_t}$, then $z_t$ does not appear after time $t$ until phase $r_{m_t} + 1$ or later. Applying Lemma~\ref{lemma:lengthtoerror} to each of the $m_t$ chains and summing, we get $$m_t \Err_t(z_t) + \sum_{i=1}^{m_t} \Err_{L(e_i, t_i)} (e_i) \geq k\frac{m_t(m_t-1)}{2} + \sum_{i=1}^{m_t} N^*(c_i) \geq \sum_{i=1}^{m_t} \frac{m_t + 1}{2} N^*(c_i).$$
Dividing both sides by $(m_t+1)/2$, we get that $$2\Err_t(z_t) + \sum_{i=1}^m \Err_{L(e_i, t_i)}(e_i) \geq \sum_{i=1}^{m_t} N^*(c_i).$$
Now sum over all times $t$. By Lemma~\ref{lemma:eunique}, all $L(e_i, t_i)$ are distinct, so the corresponding errors sum to at most $\eta$. Thus, we have $$3\eta \geq \sum_{c} N^*(c_i)$$ summing over all chains $(c,e,\dots)$ in which $e$ reappears after eviction. 
\end{proof}

Finally, we must bound the number of eviction chains $C$ across all phases. In a marker-based algorithm we would have $C = L$, with one chain per clean element. But here there is one chain per non-initial element, and not all non-initial elements are clean (and vice versa). Nonetheless, we can still bound the discrepancy $C - L$ against the prediction error. The intuition for the following proof is that any ``extra'' non-initial element must have been caused by eviction of a marked element in the previous phase. But for each chain which ends by evicting a marked element, one less unmarked element is evicted. So whereas in a marker-based algorithm, every cache element which did not appear in a given phase would be evicted by the end of the phase, in this algorithm some absent elements might remain in the cache. These elements have the potential to be clean in the next phase and yet not start eviction chains. Hence, in one case an extra chain can be charged against a non-existent chain. In the other case---when the absent element does not appear in the next phase either---the extra chain can be charged against prediction error.

\begin{lemma}\label{lemma:injection}
If $C$ is the number of non-initial elements across all phases, $L$ is the number of clean elements, and $\eta$ is the total prediction error, then $\eta \geq k(C - L)/2.$
\end{lemma}
\begin{proof}
Suppose some non-initial element $g$ in phase $r+1$ is not clean. Then it appeared in phase $r$, but was not in the cache at the end of phase $r$. In the last eviction of $g$ during phase $r$, $g$ was already marked. So by the algorithm design, that eviction must be third in some eviction chain $(c, e, g)$, where $c$ is some non-initial element in phase $r$. Then $f(c)$ does not appear in phase $r$. Nonetheless, $f(c)$ is not in $c$'s eviction chain, since $g$ is the unique element in $c$'s eviction chain which does not reappear in phase $r$. Furthermore, by definition of $f$, no other chain in phase $r$ evicts $f(c)$, so it is in the cache at the end of phase $r$. There are two cases:
\begin{enumerate}
\item $f(c)$ appears in phase $r+1$. Then $f(c)$ is clean and initial in phase $r+1$.
\item $f(c)$ does not appear in phase $r+1$. Let $r'$ be the next phase in which $f(c)$ does appear. Then if $t$ is the arrival time of $c$, by Lemma~\ref{lemma:lengthtoerror} $$\Err_{L(f(c),t)}(f(c)) + \Err_{L(e,t)}(e) \geq k(r' - r - 1).$$
\end{enumerate}

The composed map $g \mapsto c \mapsto f(c)$ is injective for fixed phase $r$, since $g$ is determined by the eviction chain of $c$, and $f$ is injective for a fixed phase. If case (2) never occurred, then every non-initial, non-clean element in phase $r+1$ would be injectively mapped to an initial, clean element in phase $r+1$, so we would have $C \leq L$. More generally, case (2) occurs for at least $C - L$ chains. Each case (2) chain provides an inequality bounding two prediction errors by at least $k$, and ideally we would simply add up the inequalities to bound $\eta$. However, some predictions may be counted in several of the inequalities. 

By Lemma~\ref{lemma:eunique}, all terms $\Err_{L(e,t)}(e)$ in the inequalities are contributed by different predictions---i.e. adding up those error terms does not double-count.

Consider a fixed time $t$. Consider the set of case-(2) chains $(c_i, e_i, g_i)$ such $L(f(c_i), t_i) = t$, where $t_i$ is the time of $c_i$'s arrival. Each chain is in a different phase, since $f$ is injective for any one phase. If there are $m_t$ such chains in phases $r_1 < \dots < r_{m_t}$, then $z_t$ does not appear after time $t$ until phase $r_{m_t} + 2$ or later. So the case-(2) inequality applied to the earliest chain $(c_1, e_1, g_1)$ gives $$\Err_t(z_t) + \Err_{L(e_1,t_1)}(e_1) \geq k(r_{m_t} + 1 - r_1) \geq km_t.$$ 

Summing the above inequality over all times $t$, each prediction error is counted at most twice---once as the first term and once as the second---whereas $\sum_t m_t \geq C - L$ as previously shown. Hence, $2\eta \geq k(C - L)$, as desired.
\end{proof}

With the above error/performance bounds in hand, we can prove a bound on the competitive ratio of the algorithm.

\begin{theorem}
The algorithm \textsc{lnonmarker} achieves a competitive ratio of $$O\left(1 + \frac{\eta/\textsc{opt}}{k} \log k\right)$$ when the prediction error is no more than $\eta$.
\end{theorem}

\begin{proof}
Fix any phase $r$. Then for any execution $E_r$ of the first $r-1$ phases, $$\mathbb{E}[\text{\# cache misses in phase $r$} | \mathcal{E}_r = E_r] = \sum_{(c,e,\dots) \in r} \mathbb{E}[\text{length}(c) | \mathcal{E}_r = E_r],$$ summing over chains in phase $r$.

By Claim~\ref{claim:chainlength}, the right hand side is bounded by $$\alpha C_r + \alpha\frac{\log k}{k} \mathbb{E}\left[\sum_{(c,e,\dots)\in r} N^*(c)\middle| \mathcal{E}_r = E_r \right],$$ where $C_r$ is the number of chains in phase $r$, and $\alpha$ is a constant. Applying law of total expectation and then summing over all phases, $$\mathbb{E}[\text{\# cache misses}] \leq \alpha\mathbb{E}[C] + \alpha\frac{\log k}{k} \mathbb{E}\left[\sum_{(c,e,\dots)} N^*(c) \right].$$

From Lemma~\ref{lemma:injection}, the inequality $C \leq 2\eta/k + L$ holds for all executions, and therefore in expectation. From Lemma~\ref{lemma:totalerror}, the inequality $\sum_c N^*(c) \leq 3\eta$ holds for all executions, where $c$ ranges over chains $(c, e, \dots)$ where $e$ reappears after eviction. Recall that $N^*(c)$ was defined to be $0$ for all other chains. So we get that

\begin{align*}
    \mathbb{E}[\text{\# cache misses}]
    &\leq 2\alpha\frac{\eta}{k} + \alpha L + 3\alpha\frac{\eta}{k}\log k.
\end{align*}

Recalling from Lemma~\ref{lemma:fiat} that $L/2 \leq \textsc{opt}$, the competitive ratio follows.
\end{proof}

The above algorithm has one flaw: it does not satisfy any robustness guarantee (at least, that has been proven), since as $\eta/\textsc{opt} \to \infty$ the bound on the number of chains disappears, and the competitive ratio becomes potentially unbounded. This can be resolved (albeit in a somewhat unsatisfactory manner) by appealing to a black-box simulation theorem. For concreteness, we recall the following theorem:

\begin{theorem}\label{theorem:sim}\cite{LV2018}
Let $A, B$ be algorithms for the caching problem with competitive ratios of $\alpha$ and $\beta$ respectively. Then there is a black box algorithm ALG with a competitive ratio of $9\min(\alpha, \beta)$.
\end{theorem}

The black box algorithm in the above theorem proceeds by simulating $A$ and $B$ and switching between them whenever one starts to heavily outperform the other. The proof generalizes without change to the learned caching problem. Since there is an $O(\log k)$-competitive algorithm for learned caching which simply ignores the predictions, Theorem~\ref{theorem:sim} implies that we can obtain an $O(\log k)$ worst-case guarantee for our predictive caching algorithm with only an extra constant factor loss:

\begin{corollary}
There is an algorithm for caching with predictions that achieves a competitive ratio of $$O\left(1 + \min\left(1, \frac{\eta/\textsc{opt}}{k}\right) \log k\right).$$
\end{corollary}

Tracing through the proof, it turns out that the exact bound is as follows:

\begin{corollary}
There is an algorithm for caching with predictions that achieves a competitive ratio of $$9\min\left(4 + 7\frac{\eta/\textsc{opt}}{k} + 3\frac{\eta/\textsc{opt}}{k} H(k), 2H(k) \right).$$
\end{corollary}

For example, as $\eta/\textsc{opt} \to 0$, the competitive ratio of \textsc{lnonmarker} approaches $4$, and so the competitive ratio of this black box algorithm approaches $36$.

\section{Lower bound}\label{section:lbound}

In this section we provide a lower bound against predictive caching algorithms. The basic strategy is to construct a distribution of inputs and predictions such that the relative prediction error $\eta/\textsc{opt}$ is not too high, but any deterministic algorithm which has access to those predictions suffers a large number of cache misses in expectation, relative to $\textsc{opt}$. Yao's minimax principle then implies a lower bound against the competitive ratios of randomized algorithms at that level of relative prediction error.

More specifically, the input distribution and predictions will be chosen such that any prefix of the input completely determines the state of the algorithm. Furthermore, each prefix is sufficiently independent of future inputs that the algorithm can essentially do no better than (a) keeping previously-seen elements in the cache (for the duration of the phase), and (b) guessing arbitrarily for the remaining unseen elements.

Fix $k$ and $n$. Let $1 \leq t \leq k$ be picked later; it is a free parameter which will determine the relative prediction error $\eta/\textsc{opt}$. Let $\Omega$ be the set of sequences that can be constructed in the following manner. Let $C_1 = \{1, \dots, t\}$ and let $S_1 = \{t+1, \dots, k\}$. For all $2 \leq r \leq n$, let $C_r = [k+t] \setminus (C_{r-1} \cup S_{r-1})$ and let $S_r$ be an arbitrary subset of $C_{r-1} \cup S_{r-1}$ of size $k-t$. Then each sequence of $\Omega$ is constructed as the concatenation of $n$ \emph{phases}, where phase $r$ consists of $3k \log k$ elements of $C_r \cup S_r$ (possibly with some omissions and necessarily with some repetitions), followed by a single copy of $C_r \cup S_r$ in increasing order, without repetitions. 

That is, each phase has length $m = 3k \log k + k$. Each phase $r$ has $t$ clean elements $C_r$ (which did not appear in the previous phase) and $k-t$ stale elements $S_r$ (which did appear). For any fixed $r>1$, a uniformly random sequence of $\Omega$, conditioned on $C_r$, has a uniformly random set of stale elements $S_r$. Furthermore, conditioned on $C_r$ and $S_r$, each of the first $3k \log k$ elements of phase $r$ is independent and uniformly distributed over $C_r \cup S_r$.

We must also define the predictions for a fixed sequence. For the first $3k \log k$ elements of phase $r$, each prediction is the subsequent timestep (i.e. $h(z_i) = i+1$). For the final copy of $C_r \cup S_r$, each prediction is the end of phase $r+1$.

\begin{claim}\label{claim:error}
The overall prediction error is $O(nk^2 \log k)$, where $n$ is the number of phases in $\Omega$.
\end{claim}

\begin{proof}
Fix a phase $r$. Among the first $3k\log k$ elements of the phase, there are at most $k$ distinct elements. For each, the prediction error telescopes to at most $3k \log k + k$. For each of the $k$ elements in the final copy of $C_r \cup S_r$, the true next arrival is either in phase $r+1$ or phase $r+2$, so the error is at most $O(k \log k)$. Thus, the error in phase $r$ is $O(k^2 \log k)$.
\end{proof}

Now we would like to lower bound the average cache misses of any algorithm on $\Omega$. We must first give some simplifying notation and a probabilistic lemma:

For any time $T$ in phase $r$, let $\sim_T$ be the equivalence relation on $\Omega$ where sequences $z$ and $z'$ are equivalent if $z_j = z_j'$ for $j \leq T$. Observe that $z \sim_T z'$ implies that $h(z_j) = h(z'_j)$ for all $j \leq T$ by how the predictions were constructed. Thus, for any equivalence class $[z]_T$ of $\sim_T$, the algorithm has identical executions on sequences in $[z]_T$ up to time $T$.

For $z \in \Omega$ and time $T$, let $\mathcal{C}_{z,T}$ be the cache at time $T$ on input $[z]_T$.

\begin{lemma}\label{lemma:boundedgeom}
Fix $k$, $l$ with $2 \leq l \leq k$. Let $X_1,\dots,X_{3k\log k}$ be independent random variables uniformly distributed over $[k]$. Let $Y_i$ be the number of distinct elements in $\{X_1,\dots,X_i\} \cap [l]$ and for $0 \leq j < l$ let $T_j$ be the number of $i$ such that $Y_i = j$. Then $\mathbb{E} [T_j] \geq k/(l-j) - 1/k$.
\end{lemma}

\begin{proof}
Extend the sequence $X$ to an infinite sequence. Extend $Y$ accordingly, and define $\hat{T}_j$ as the number of steps in the infinite sequence at which $Y_i = j$. For $0 \leq j < l$, let $S_j = T_0 + \dots + T_j$ and $\hat{S}_j = \hat{T}_0 + \dots + \hat{T}_j$. Then $S_j = \min(\hat{S}_j, 3k\log k)$, so 
\begin{align*}
\mathbb{E} [T_j] 
&\geq \mathbb{E} [\min(\hat{S}_j, 3k\log k)] - \mathbb{E} [\hat{S}_{j-1}] \\
&= \mathbb{E} [\hat{T}_j] - \mathbb{E}[(\hat{S}_j - 3k\log k)\mathbbm{1}_{\hat{S}_j \geq 3k\log k}].
\end{align*}

For the first term, $\hat{T}_j$ is a geometric random variable and $\mathbb{E} [\hat{T}_j] = k/(l-j)$. For the second, $$\mathbb{E}[(\hat{S}_j - 3k\log k) \mathbbm{1}_{\hat{S}_j \geq 3k\log k}] \leq \sum_{c=3k\log k}^\infty \Pr[\hat{S}_j \geq c].$$

Observe that for any $c$, if $\hat{S}_j \geq c$ then $Y_c \leq j < l$, which occurs with probability at most $l \left(1 - 1/k\right)^c.$ It follows that $$\mathbb{E}[(\hat{S}_j - 3k\log k)\mathbbm{1}_{\hat{S}_j \geq 3k\log k}] \leq \sum_{c = 3k\log k}^\infty l\left(1 - \frac{1}{k}\right)^c \leq \frac{1}{k}.$$ We conclude that $\mathbb{E} [T_j] \geq k/(l-j) - 1/k$.
\end{proof}

\begin{theorem}\label{theorem:lbound}
For any deterministic algorithm with access to the predictions $h$, the expected number of cache misses achieved on an input sampled uniformly at random from $\Omega$ is at least $$O\left(n t \log \frac{k}{t} \right).$$
\end{theorem}

\begin{proof}
There are $n$ phases in $\Omega$. Fix a phase $r$. We claim that the expected number of cache misses in phase $r$ is at least $O(t \log k/t)$.

Recall that each phase has length $m = 3k\log k + k$. Let $T = (r-1)m$: the final index of phase $r-1$. Fix $z \in \Omega$. Fix $1 \leq i \leq k \log k$ and consider any class $[w]_{T+i} \subseteq [z]_T$. Picking a sequence $W \in [w]_{T+i}$ uniformly at random, the clean elements $C_r$ and the cache $\mathcal{C}_{W,T+i}$ are determined, but $S_r$ is a random variable. Let $S_\text{seen}$ be the set of elements which have already been seen in phase $r$ by time $T+i$, excluding $C_r$. Then $S_r$ must contain $S_\text{seen}$, but $S_r \setminus S_\text{seen}$ is a uniformly random subset of $[k+t] \setminus (C_r \cup S_\text{seen})$ of size $k - t - |S_\text{seen}|$. Hence, 
\begin{align*}
\Pr_W[\text{cache hit on $W_{T+i}$}]
&= \sum_{\hat{S}} \Pr_W[S_r = \hat{S}] \Pr_W[\text{cache hit on $W_{T+i}$}| S_r = \hat{S}] \\
&= \sum_{\hat{S}} \Pr_W[S_r = \hat{S}] \frac{|\mathcal{C}_{W,T+i} \cap (C_r \cup \hat{S})|}{k} \\
&= \frac{|\mathcal{C}_{W,T+i} \cap (C_r \cup S_\text{seen})| + \mathbb{E}_W |\mathcal{C}_{W,T+i} \cap (S_r \setminus S_\text{seen})|}{k}
\end{align*}
If $|\mathcal{C}_{W,T+i} \cap (C_r \cup S_\text{seen})| = a$, then $\mathcal{C}_{W,T+i}$ contains $k - a$ elements in $[k+t] \setminus (C_r \cup S_\text{seen})$, each of which is contained in $S_r \setminus S_\text{seen}$ with probability $(k - t - |S_\text{seen}|)/(k - |S_\text{seen}|)$. So the above expression is maximized when $a = |C_r \cup S_\text{seen}| = k+t$. Thus,
$$\Pr_W[\text{cache hit on $W_{T+i}$}] \leq \frac{t + |S_\text{seen}| + (k - t - |S_\text{seen}|) \cdot \frac{k - t - |S_\text{seen}|}{k - |S_\text{seen}|}}{k}.$$
Simplifying, it follows that
$$\Pr_W[\text{cache miss on $W_{T+i}$}] \geq \frac{t}{k} \frac{k - t - |S_\text{seen}|}{k - |S_\text{seen}|}.$$
If $Z \in [z]_T$ is chosen uniformly at random, then for any $0 \leq N < k-t$, the expected number of times at which $|S_\text{seen}| = N$ is at least $k/(k-t-N) - 1/k \geq k/(2(k - t - N))$ by Lemma~\ref{lemma:boundedgeom}. Summing over $N$, the expected number of cache misses in phase $r$ on input $Z$ is at least $$\sum_{N=0}^{k-t-1} \frac{t}{k} \frac{k-t-N}{k-N} \frac{k}{2(k-t-N)} = \sum_{N=0}^{k-t-1} \frac{t}{2(k-N)} \geq \frac{t}{2}\left(\frac{1}{t+1} + \dots + \frac{1}{k}\right) = \Omega(t \log k/t).$$ Since the equivalence classes of $\sim_T$ partition $\Omega$, and the above bound holds for each class, it holds for $\Omega$ as desired.
\end{proof}

Now for any fixed ``relative prediction error'' $\epsilon$ with $k \log k \leq \epsilon \leq k^2 \log k$, we can pick $t = (k^2 \log k)/\epsilon$. Then every sequence in $\Omega$ has $nt$ clean elements, and thus $\textsc{opt} = \Theta(nt)$. Furthermore, by Claim~\ref{claim:error}, we have $\eta \leq nk^2 \log k$. Thus, $\eta/\textsc{opt} \leq \epsilon$.

 But by Theorem~\ref{theorem:lbound}, any deterministic algorithm requires $nt \log k/t$ cache misses in expectation on $\Omega$. By Yao's minimax principle, for any randomized algorithm there is some input $z \in \Omega$ for which the algorithm incurs $nt \log k/t$ cache misses in expectation. Hence, we have the following result:
 
\begin{theorem}
Let $A$ be a randomized online algorithm for caching, which has access to next-arrival predictions. For any $\epsilon$, the algorithm achieves competitive ratio no better than $$\Omega\left(\log \min\left(\frac{\epsilon}{k \log k}, k \right)\right)$$ when restricted to inputs with $\eta/\textsc{opt} \leq \epsilon$.
\end{theorem}

\paragraph{Acknowledgments.} I would like to thank Costis Daskalakis and Piotr Indyk for introducing me to the area of learning-augmented algorithms, and for their advice and encouragement. 

\end{document}